\documentclass[11pt]{article}
\usepackage{amssymb,amsmath,amsthm,bbm,verbatim}
\usepackage{graphicx, color}

\usepackage{geometry}
\usepackage{setspace}
\usepackage{enumerate}

\newcommand{\XX}[1]{#1}

\newtheorem{theorem}{Theorem}
\newtheorem{lemma}[theorem]{Lemma}
\newtheorem{fact}{Fact}

\newtheorem{claim}[theorem]{Claim}

\newtheorem{corollary}[theorem]{Corollary}

\theoremstyle{remark}

\newtheorem{definition}[theorem]{Definition}

\newcommand{\amir}[1]{}

\newcommand{\R}{\mathbb R}
\newcommand{\D}{{\ensuremath{\mathcal{D}}}}
\newcommand{\intr}[1]{\forall \ #1 \ } 
\newcommand{\andop}[1]{\ensuremath{\mathop{\wedge}_{#1}}}

\newcommand{\G}{\mathcal{G}}
\newcommand{\U}{\ensuremath{\mathcal{U}}}
\newcommand{\E}{\mathbb E}

\newcommand{\hh}{\ensuremath{\mathcal{H}}}
\newcommand{\uh}{\ensuremath{\mathcal{U}}}
\newcommand{\gcr}{\ensuremath{\mathcal{G}}}
\newcommand{\gmr}{\ensuremath{\mathcal{G_{MR}}}}
\newcommand{\zo}{\ensuremath{\{0,1\}}}

\newcommand{\sig}{\ensuremath{\bar{\sigma}}}
\newcommand{\abs}[1]{\ensuremath{\left| #1 \right|}}

\newcommand{\rgta}{\rightarrow}

\newcommand{\Var}{\mathsf{Var}}
\newcommand{\poly}{\mathrm{poly}}
\newcommand{\nm}{\mathrm{L_1}}

\newcommand{\eps}{\varepsilon}

\newcommand{\mult}{\mathsf{mult}}

\newcommand{\prg}{{\ensuremath \mathsf{PRG}}}

\newcommand{\ind}{\ensuremath{\mathbbm{1}}}

\newcommand{\eat}[1]{}

\begin{document}

\makeatletter

\makeatother
\author{
\and Parikshit Gopalan\thanks{
Microsoft Research.
{\tt parik@microsoft.com}.}
\and
Amir Yehudayoff\thanks{Department of Mathematics, Technion--IIT.
{\tt amir.yehudayoff@gmail.com}.
Horev fellow -- supported by the Taub foundation.  
Research supported by ISF and BSF}
}

\title{Inequalities and tail bounds for elementary symmetric polynomials with applications}

\date{}
\begin{titlepage}
\thispagestyle{empty}
\maketitle

\begin{abstract}

We study the extent of independence needed to approximate the product
of bounded random variables in expectation, a natural question  that
has applications in pseudorandomness and min-wise independent hashing.

For random variables whose absolute value is bounded by $1$, we give
an error bound of the form $\sigma^{\Omega(k)}$ where $k$ is the
amount of independence and $\sigma^2$ is the total variance of the
sum. Previously  known bounds only applied in more restricted
settings, and were quanitively weaker. We use this to give a simpler
and more modular analysis of a construction of min-wise independent
hash functions and pseudorandom generators for combinatorial
rectangles due to Gopalan {\em et al.}, which also slightly improves
their seed-length. 
  
Our proof relies on a new analytic inequality for the elementary symmetric
  polynomials $S_k(x)$ for $x \in \R^n$ which we believe to be of
  independent interest. We show that if
  $|S_k(x)|,|S_{k+1}(x)|$ are small relative to $|S_{k-1}(x)|$ for
  some $k>0$ then $|S_\ell(x)|$ is also small for all $\ell > k$. 
  From these, we derive tail bounds for the elementary symmetric
  polynomials when the inputs are only $k$-wise independent.

\end{abstract}

\end{titlepage}

\section{Introduction}

The power of independence in probability and
randomized algorithms stems from the fact that it lets us
control expectations of products of random variables. If
$X_1,\ldots,X_n$ are independent random variables, then
$\E[\prod_{i=1}^nX_i] = \prod_{i=1}^n\mu_i$ where $\mu_i$ are their
respective means. However, there are numerous settings in computer
science, where true independence either does not hold, or is too
expensive (in terms of memory or randomness).
Motivated by this, we explore settings when approximate versions of
the product rule for expectations hold  even with limited independence.
Concretely, let $X_1,\ldots,X_n$ be random variables lying in the range
$[-1,1]$, with mean $\mu_i$ and variance $\sigma^2_i$ repectively.
We are interested in the smallest $k = k(\delta)$ such that whenever the $X_i$s are
drawn from a $k$-wise independent distribution $\D$,  
it holds that
\begin{align}
\label{eq:k-wise}
\left|\E_\D[\prod_{i=1}^n X_i] - \prod_{i=1}^n\mu_i \right| \leq \delta. 
\end{align}

As stated, we cannot hope to make do even with $k = n-1$. Consider the
case where each $X_i$ is a random $\{\pm 1\}$ bit. If $X_n =\prod_{i \leq
  n-1}X_i$, the resulting distribution is $(n-1)$-wise independent, but
$\E[\prod_iX_i] =1$, whereas it is $0$ with true independence. So,
clearly, we need some additional assumptions about the random
variables. 

The main message of this paper is that small total variance is
sufficient to ensure that the product rule holds approximately even under
$k$-wise independence.  

\begin{theorem}
\label{thm:intro-product}
Let $X_1,\ldots,X_n$ be random variables each distributed in the range
$[-1,1]$, with mean $\mu_i$ and variance $\sigma^2_i$ repectively. Let
$\sigma^2 = \sum_{i} \sigma_i^2$. 
There exist  constants $c_1 >1$ and $1 > c_2 > 0$ such that under any $k$-wise independent distribution $\D$, 
\begin{align}
\label{eq:k-wise-sigma}
\left|\E_\D[\prod_{i=1}^n X_i] - \prod_{i=1}^n\mu_i \right| \leq (c_1\sigma)^{c_2k}. 
\end{align}
Specifically, if $\sigma < 1/(2c_1)$ then $k = O(\log(1/\delta)/\log(1/\sigma))$-wise
independence suffices for Equation~\eqref{eq:k-wise}.  
\end{theorem}

An important restriction that naturally arises is positivity, where each $X_i$ lies in the interval $[0,1]$. 
This setting of parameters (positive variables, small total variance) is important for the applications
considered in this paper: pseudorandom generators for combinatorial
rectangles \cite{EGLNV,LinialLSZ97} and min-wise independent permutations
\cite{BCFM}. The former is an important
problem in the theory of unconditional pseudorandomness
which has been studied intensively \cite{EGLNV,LinialLSZ97,SaksSZZ99,ArmoniSWZ96,Lu02,GMRTV}. 
Min-wise independent hashing was introduced by Broder et
al.~\cite{BCFM} motivated by similarity estimation, and
further studied by \cite{Indyk, BCM, SaksSZZ99}. \cite{SaksSZZ99}
showed that PRGs for rectangles give min-wise independent hash
functions.

The results of \cite{EGLNV,Indyk} tell us that under $k$-wise
independence, positivity and boundedness, 
the LHS of Equation~\eqref{eq:k-wise} is bounded by
$\exp(-\Omega(k))$, hence $k = O(\log(1/\delta))$ suffices for error
$\delta$. In contrast, we have seen that such a bound cannot hold in 
the $[-1,1]$ case.
However, once the variance is smaller than some
constant, our bound beats this bound even in the $[0,1]$ setting. Concretely, when $\sigma^2 < n^{-\eps}$ for some $\eps
>0$, our result says that $O(1)$-wise independence suffices for
inverse polynomial error in Equation~\eqref{eq:k-wise}, as opposed to
$O(\log(n))$-wise independence.  
This improvement is crucial in analyzing PRGs and hash functions in the polynomially small error regime.
A recent result of \cite{GMRTV} achieves near-logarithmic seed-length for both these problems,
even in the regime of inverse polynomial error. Their construction is
simple, but its analysis is not.
Using our results, we give a modular analysis of the pseudorandom
generator  construction for rectangles of \cite{GMRTV}, using the viewpoint of hash
functions. 

Our analysis is simpler and perhaps more intuitive.
It also improves the seed-length of the construction, getting the
dependence on the dimension $n$ down to $O(\log\log(n))$ as opposed to
$O(\log(n))$, which (nearly) matches a lower bound due to \cite{LinialLSZ97}.  
Given the basic nature of the question, we feel our
results might find other applications. Very recently, \cite{GKM15} constructed
the first pseudorandom generators with near-logrithmic seed-length
for several classes of functions including
halfspaces, modular tests and combinatorial shapes. The key technical
ingredient of their work is a generalization of Theorem \ref{thm:intro-product}
to the setting where each $X_i$ takes values in the unit complex disc. 

The main technical ingredient in our work is a new analytic inequality about 
symmetric polynomials in real variables which we believe is
independently interesting. The $k$'th symmetric polynomial in $a = (a_1,a_2,\ldots,a_n)$ is defined as
\begin{align}
\label{eq:sk} 
S_k(a) =   \sum_{T \subseteq [n]: |T|=k} \prod_{i \in T} a_i 
\end{align}
(we let $S_0(a) =1$). We show that for any real vector $a$, if $|S_k(a)|,|S_{k+1}(a)|$ are
small relative to $|S_{k-1}(a)|$ for some $k>0$, then $|S_\ell(a)|$ is
also small for all $\ell > k$. This strengthens and generalzies a
result of \cite{GMRTV} for the case $k=1$. 

We give an overview of the new inequality, its use in the
derivation of bounds under limited independence, and finally the application of these bounds
to the construction of pseudorandom generators and hash functions. 

\subsection{An inequality for elementary symmetric polynomials}

The elementary polynomials appear as coefficients of a 
univariate polynomial with real roots, since $\prod_{i \in [n]} (\xi +
a_i) = \sum_{k=0}^n \xi^k S_{n-k}(a)$. 
Symmetric polynomials have been well studied in
mathematics, dating back to classical results of Newton and
Maclaurin (see \cite{CS} for a survey).
This work focuses on their growth rates. Specifically, we study how local information on $S_k(a)$
for two consecutive values of $k$ implies global information for all larger values of $k$.

It is easy to see that symmetric polynomials over the real numbers have the following property:
\eat{
We may write
\begin{align}
\label{eqn:propOfSk}
p(\xi) = \prod_{i \in [n]} (\xi b_i + 1) = \sum_{k=0}^n \xi^k S_k(b).
\end{align}
The condition on the derivatives of $p$
is equivalent to $S_1(b) = S_2(b) = 0$,
and the following fact completes the argument.}
\begin{fact} 
\label{fact;1}
Over the real numbers, if $S_1(b)=S_2(b) = 0$ then $b=0$.
\end{fact}
This is equivalent to saying that if $p(\xi)$ is a real univariate polynomial
of degree $n$ with $n$ nonzero roots and $p'(0) = p''(0) = 0$ then $p \equiv 0$.
This does not hold over all fields, for example, the polynomial
$p(\xi) = \xi^3+1$ has three nonzero  complex roots and $p'(0) = p''(0) = 0$.

A robust version of Fact~\ref{fact;1} was recently proved in \cite{GMRTV}:
For every $a \in \R^n$ and $k \in [n]$,
\begin{align}
\label{eqn:GMR}
 |S_k(a)| \leq  \left(S^2_1(a) + 2|S_2(a)|\right)^{k/2} .
\end{align}
That is, if $S_1(a),S_{2}(a)$ are small in absolute value, then so is
everything that follows. We provide an essentially optimal bound.

\begin{theorem}
\label{thm:boundUsingS1E2}
For every $a \in \R^n$ and $k \in [n]$,
\[ |S_k(a)| \leq  \left(\frac{6 e (S^2_1(a) + |S_2(a)|)^{1/2}}{k^{1/2} }\right)^k .\]
\end{theorem}

The parameters promised by Theorem~\ref{thm:boundUsingS1E2} are tight up to an exponential in $k$
which is often too small to matter
(we do not attempt to optimise the constants).
For example, if $a_i = (-1)^i$
for all $i  \in [n]$ then
$|S_1(a)| \leq 1$ and $|S_2(a)| \leq n+1$
but
$S_k(a)$ is roughly $(n/k)^{k/2}$.

A more general statement than Fact~\ref{fact;1} actually holds (see
Appendix \ref{app:proof} for a proof).
\begin{fact}
\label{fact:2}
Over the reals, if $S_k(a) = S_{k+1}(a) = 0$ for $k > 0$
then $S_\ell(a) =0$ for all $\ell \geq k$. 
\end{fact}
We prove a robust version of this fact as well:
A twice-in-a-row bound on the increase of the symmetric functions
implies a bound on what follows.

\begin{theorem}
\label{thm:GenBound}
For every $a \in \R^n$, if $S_k(a) \neq 0$ 
and 
\begin{align*}
\left|{k+1 \choose k}\frac{S_{k+1}(a)}{S_k(a)} \right| & \leq C
\ \ \text{and} \ \
\left| {k+2 \choose k}\frac{S_{k+2}(a)}{S_{k}(a)} \right| \leq
C^2
\end{align*}
then for every $1 \leq h \leq n -k$,
\begin{align*}
\left|{ k + h \choose k} \frac{S_{k+ h}(a)}{S_k(a)}\right| \leq \left(\frac{6eC}{h^{1/2}}\right)^h .
\end{align*}
\end{theorem}

Theorem~\ref{thm:GenBound} is proved by reduction to Theorem~\ref{thm:boundUsingS1E2}.
The proof of Theorem~\ref{thm:boundUsingS1E2} is analytic and uses the
method of Lagrange multipliers, and is different from that of
\cite{GMRTV} which relied on the Newton-Girrard identities.  The
argument is quite general, and similar bounds may be obtained for
functions that are recursively defined. 
\XX{The proof can be found in Section~\ref{sec:inequalities}.}

Stronger bounds are known when the inputs are nonnegative.
When $a_i \geq 0$ for all $i \in [n]$, the classical Maclaurin inequalities \cite{CS} imply that
$S_k(a) \leq (e/k)^k(S_1(a))^k$.
In contrast, when we do not assume non-negativity, one cannot hope for
such bounds to hold under the assumption that $|S_1(a)|$ or any single
$|S_k(a)|$ is small (cf.\ the alternating signs example above).

\eat{it is carried by applying
Theorem~\ref{thm:boundUsingS1E2}
to the $k$'th derivative of the polynomial defined in Equation \eqref{thm:boundUsingS1E2}.
The proof can be found at the end of Section~\ref{sec:inequalities}.}

\subsection{Expectations of products under limited independence} 

\eat{Pseudorandomness studies the possibility making efficient computation
deterministic with only small losses in efficiency.}

We return to the question alluded to earlier about how much independence is
required for the approximate product rule of
expectation.  This question arises in the context of min-wise hashing
\cite{Indyk}, PRGs for combinatorial rectangles
\cite{EGLNV, GMRTV}, read-once DNFs \cite{GMRTV} and more.

%

One could derive bounds of similar shape to ours using the work of \cite{GMRTV},
but with much stronger assumptions on the variables. 
More precisely, one would require  $\E[X_i^{2k}] \leq
(2k)^{2k}\sigma_i^{2k}$ for all $i \in [n]$, and get an error
bound of roughly $k^{O(k)}(\sum_i \sigma_i^2)^{\Omega(k)}$. These stronger assumptions limit
the settings where their bound can be applied (biased variables
typically do not have good moment bounds), and ensuring these conditions hold
led to tedious case analysis  in analyzing their PRG construction.

We briefly outline our approach.
We start from the results of \cite{EGLNV,Indyk} who give an error bound of $\exp(-k)$.
To prove this, they consider random
variables $Y_i = 1 -X_i$, so that
\begin{align}
\label{eq:incl-excl} 
\prod_{i=1}^nX_i = \prod_{i=1}^n(1- Y_i)  = \sum_{j=0}^n(-1)^jS_j(Y_1,\ldots,Y_n). 
\end{align}
By inclusion-exclusion/Bonferroni inequalities, the series on the
right gives alternating upper and lower bounds, and the error incurred by
truncating to $k$ terms is bounded by $S_k(Y)$. So we can bound the expected error by $\E[S_k(Y)]$ for which  $k$-wise
independence suffices.

Our approach replaces inclusion-exclusion by a Taylor-series style expansion about the mean, as in \cite{GMRTV}. Let us assume $\mu_i \neq 0$ and let
$X_i = \mu_i(1 + Z_i)$. Thus,
\begin{align}
\label{eq:taylor}
\prod_{i=1}^nX_i = \prod_{i=1}^n\mu_i(1+ Z_i) = \prod_{i=1}^n\mu_i\left( \sum_{j=0}^nS_j(Z)\right).
\end{align}
If this series were alternating, then we would
only need to bound $\E[|S_k(Z)|]$, which is easy. However, this need
not be true since $Z$ may have negative entries (even if we start with
$X_i$s all positive). So, to
argue that the first $k$ terms are a good approximation, we need to
bound the tail $\sum_{\ell \geq k}S_\ell(Z)$. At first, this seems
problematic, since this involves high degree polynomials, and it seems
hard to get their expectations right assuming just $k$-wise
independence\footnote{We formally show this in Section \ref{sec:lower}.}.
Even though we cannot bound $\E[S_\ell(Z)]$ under
$k$-wise independence once $\ell \gg k$, we use our new
inequalities for symmetric polynomials to get strong tail bounds on
them. This lets us show that truncating Equation~\eqref{eq:taylor} after $k$ terms gives error roughly $O(\sigma^{ck})$, and thus $k =
O(\log(1/\delta)/\log(1/\sigma)$ suffices for error $\delta$.
We next describe these tail bounds in detail.

We assume the following setup: $Z =(Z_1,\ldots,Z_n)$ is a vector of real valued random variables
where $Z_i$ has  mean $0$ and variance $\sigma_i^2$, and $\sigma^2
= \sum_{i} \sigma_i^2 < 1$.  Let $\cal{U}$ denote the distribution where the coordinates of $Z$ are independent.
\eat{
It is easy to show (see Lemma \ref{lem:exp-kwise}) that \[\E_{X \in \cal{U}}[|S_\ell(X)|] \leq \frac{\sigma^\ell}{\sqrt{\ell!}}.\]
In particular, if $\sigma^2 <1$ then $\E[|S_\ell|]$
decays exponentially with $\ell$.  For $t >0$ and $t\sigma
\leq 1/2$, we may also conclude }
One can show that $\E_{Z \in \U}[|S_\ell(Z)|] \leq
\sigma^\ell/\sqrt{\ell!}$ and hence by Markov's inequality  (see Corollary \ref{cor:k-wise}) when $t >1$ and $t\sigma \leq 1/2$, 
\begin{align}
\label{eq:tail}
\Pr_{Z \in \U}\left[\sum_{\ell = k}^n |S_\ell(Z)| \geq 2(t\sigma)^k\right]
\leq 2t^{-2k}.
\end{align}

Although $k$-wise independence does not suffice to bound
$\E[S_\ell(Z)]$ for $\ell \gg k$,  we use Theorem \ref{thm:GenBound} to
show that a similar tail bound  holds under limited independence.

\begin{theorem}
\label{thm:main}
Let $\cal{D}$ denote a distribution over $Z = (Z_1,\ldots,Z_n)$
as above where the $Z_i$s are $(2k+2)$-wise independent.
For $t >0$ and\footnote{A weaker but more technical assumption
on $t,\sigma,k$ suffices, see Equation \eqref{eqn:Endsum}.} $16 e t\sigma \leq 1$,
\begin{align}
\label{eq:kwise-tail2}
\Pr_{X \in \cal{D}} \left[\sum_{\ell=k}^n |S_\ell(Z)| \geq 2(6 e t\sigma)^k \right] 
\leq 2t^{-2k} .
\end{align}
\end{theorem}

Typically proofs of tail bounds under limited independence proceed by bounding the expectation of some
suitable low-degree polynomial. The proof of Theorem \ref{thm:main}
does not follow this route. In Section \ref{sec:lower}, we give an example of $Z_i$s and a $(2k +2)$-wise independent
distribution on where $\E[|S_\ell(Z)|]$ for $\ell \in \{2k +3,\ldots, n- 2k
-3\}$ is much larger than under the uniform distribution. The same
example also shows that our tail bounds are close to tight.

\eat{Lemma~\ref{lem:main} below shows that for $t > 0$ and for $\ell
\geq k$, except with $\cal{D}$-probability $2t^{-2k}$ 
\[ |S_\ell(X)|  \leq (6 e t\sigma)^{\ell}  
\left( \frac{k}{\ell} \right)^{\ell/2}. \]}

\eat{In place of Theorem \ref{thm:boundUsingS1E2}, one could plug in the bound given by Equation
\eqref{eqn:GMR} \XX{to} get a somewhat weaker version of Theorem \ref{thm:GenBound}.  
However, it seems that the resulting bound will not be
strong enough to prove Theorem \ref{thm:main}, and the asymptotic
improvement given by Theorem \ref{thm:boundUsingS1E2} is crucial.}

\subsection{Applications to pseudorandom generators and hash functions}

A hash function is a map $h:[n] \rgta [m]$.
Let $\uh$ denote the family of all hash functions $h:[n] \rgta [m]$.
Let $\hh \subseteq \uh$ be a family of hash functions.   
For $S \subseteq [n]$, let $\min h(S) = \min_{x \in S} h(x)$. 
The notion of min-wise independent hashing was introduced by Broder et
al.~\cite{BCFM} motivated by similarity estimation, and independently
by Mulmuley~\cite{Mulmuley} motivated by computational geometry.
The following generalization was introduced by Broder et al.~\cite{BCM}:

\begin{definition}
We say that $\hh: [n] \rgta [m]$ is
approximately $\ell$-minima-wise independent with error $\eps$ if for
every $S \subseteq [n]$ and for every sequence $T = (t_1,\ldots,t_\ell)$ of $\ell$ distinct
elements of $S$,
\[ \left|\Pr_{h \in \hh}[h(t_1) < \cdots < h(t_\ell) < \min
  h(S\setminus T)] -  \Pr_{h \in \uh}[h(t_1) < \cdots < h(t_\ell) < \min
  h(S\setminus T)] \right| \leq \eps. \]
\end{definition}
 
Combinatorial rectangles are a well-studied class of tests in
pseudorandomness
\cite{EGLNV,LinialLSZ97,SaksSZZ99,ArmoniSWZ96,Lu02,GMRTV}. In addition
to being a natural class of statistical tests, constructing generators
for them with optimal seeds (up to constant factors) will improve on
Nisan's generator for logspace \cite{ArmoniSWZ96}, a long-standing
open problem in derandomization.

\begin{definition}
A combinatorial rectangle is a function $f:[m]^n \rgta \zo$ which is
specified by $n$ co-ordinate functions $f_i: [m] \rgta \zo$ as 
$f(x_1,\ldots,x_n) = \prod_{i \in m}f_i(x_i)$.
A map $\gcr:\zo^r \rgta [m]^n$ is a $\prg$ for combinatorial
rectangles with error $\eps$ if for every combinatorial rectangle $f:[m]^n \rgta \zo$, 
\[ \left|\E_{x \in \zo^r}[f(\gcr(x))] -\E_{x \in [m]^n}[f(x)]\right| \leq \eps. \]
\end{definition}

A generator $\gcr:\zo^r \rgta [m]^n$ can naturally be thought of as
a collection of $2^r$ hash functions, one for each seed. 
For $y \in \zo^r$, let $\gcr(y) = (x_1,\ldots,x_n)$. The corresponding hash
function is given by $g_y(i) = x_i$. The corresponding hash functions
have the property that the probability that they fool all test
functions given by combinatorial rectangles. Saks et al.~\cite{SaksSZZ99}
showed that this suffices for $\ell$-minima-wise independence.
They state their result for $\ell =1$, but their proof  extends to all
$\ell$ (see appendix \ref{app:proof}). Constructions of PRGs for
rectangles and min-wise hash functions that 
achieve seed-length $O(\log(mn)\log(1/\eps))$ were given by
\cite{EGLNV} and \cite{Indyk} respectively using limited independence.
The first construction $\gmr$ to achieve seed-length $\tilde{O}(\log(mn/\eps))$ was given recently
by \cite{GMRTV}. We use our results to give an anlysis of their
generator which we believe is simpler and more intutitive, which also
improves the seed-length, to (nearly) match the lower bound from
\cite{LinialLSZ97}.

We take the view of $\gmr$ as a collection of hash functions
$g:[n] \to [m]$,  based on iterative applications of an {\em
  alphabet squaring} step. We describe the generator formally  in
Section \ref{sec:gmr}. We start by observing that fooling rectangles is
easy when $m$ is small; $O(\log(1/\delta))$-wise independnce
suffices, and this requires $O(\log(1/\delta)\log(m))
=O(\log(1/\delta))$ random bits for $m =O(1)$. 

The key insight in \cite{GMRTV} is that
gradually increasing the alphabet is also easy (in that it requires only
logarithmic randomness). Assume that we have a hash function $g_0:[n] \to [m]$ and from it, we define
$g_1:[n] \to [m^2]$. To do this, we pick a function $g_1':[n]
\times [m] \to [m^2]$ and set $g_1(i) = g_1'(i,g_0(i))$. The key
observation is that it suffices to pick $g_1'$ using
only $O(\log(1/\delta)/\log(m))$-wise independence (rather than the
$O(\log(1/\delta))$-wise independence needed for one shot).

To see why this is so, fix subsets $S_i \subset [m^2]$ for each
co-ordinate and pretend that $g_0$ is truly random. One can show that
$\Pr_{g_0}[g_1(i) \in S_i]$ is a random variable over the
choice of $g_1'$ with variance $1/\poly(m)$. Since we are
interested in $\prod_{i}\Pr_{g_0}[g_1(i) \in S_i]$, which is the
product of $n$  small variance random variables, Theorem
\ref{thm:intro-product} says it suffices to
use limited independence\footnote{To optimize the seed-length, we actually use
almost $k$-wise independence rather than exact $k$-wise
independence. So the analysis does not use Theorem \ref{thm:intro-product} as
a black-box, but rather it directly uses Theorem \ref{thm:main}.}.

\begin{theorem}
\label{thm:gmr}
Let $\gmr$ be the family of hash functions from $[n]$ to $[m]$
defined in Section \ref{sec:gmr-def}  with error parameter $\delta > 0$. 
The seed length is at most $O((\log\log(n) + \log(m/\delta)) \log\log
(m/\delta))$.  Then, for every $S_1, \ldots,S_n \subseteq [m]$,
\[ \left|\Pr_{g \in \gmr}[\intr{i \in [n]}g(i) \in S_i] - \Pr_{h \in
  \uh}[\intr{i \in [n]}h(i) \in S_i]\right| \leq \delta. \]
\end{theorem}

This improves the \cite{GMRTV} bound in the dependence on $n$
and $\delta$ (their bound was $O(\log(mn/\delta)\log\log(m) +
\log(1/\delta)\log\log(1/\delta)\log\log\log(1/\delta))$). In
particular, the dependence on $n$ reduces from $\log(n)$ to
$\log\log(n)$\footnote{The reason $\log\log(n)$ seedlength is possible is because every
  rectangle can be $\eps$-approximated by one that depends only on
  $O(m\log(1/\eps))$ co-ordinates. Hence the number of functions to
  fool grows polynomially in $n$, rather than exponentially.}. \cite{LinialLSZ97} showed a lower bound of
$\Omega(\log(m) + \log(1/\eps) + \log\log(n))$ even for hitting sets,
so our bound is tight upto the $\log\log(m/\delta)$ factor. While
\cite{LinialLSZ97} constructed hitting-set generators for rectangles
with near-optimal seedlength, we are unaware of previous constructions
of pseudorandom generators for rectangles where the dependence of the seedlength on
$n$ is $o(\log(n))$.

Combining this with Theorem \ref{thm:prg-to-hash}, we get the following corollary.

\begin{corollary}
\label{cor:mwh}
For every $\ell$, there is a family of 
approximately $\ell$-minima-wise independent hash functions with error $\eps$ and 
seed length at most $O((\log\log(n) + \log(m^\ell/\eps))(\log\log(m^\ell/\eps)))$.
\end{corollary}

\subsection{Subsequent work}

Very recently, Gopalan, Kane and Meka \cite{GKM15} constructed
the first pseudorandom generators with seed-length
$O((\log(n/\delta)\log\log(n/\delta)^2)$ for several classes of functions including
halfspaces, modular tests and combinatorial shapes. The key technical
ingredient of their work is a generalization of Theorem~\ref{thm:intro-product} to the setting where the $X_i$s are complex valued
random variables lying in the unit disc. Their proof however is very
different from ours, and in particular it does not imply 
the inequalities and tail bounds for symmetric polynomials that are proved
here. 

\eat{
\paragraph{Organization:} We present the proofs of our inequalities for symmetric polynomials in the main body, since they are the main technical contribution
of this work, and defer other proofs to the appendix. 
}

\paragraph{Organization:} We present the proofs of our inequalities
for symmetric polynomials in Section \ref{sec:inequalities} and tail bounds for
symmetric polynomials in Section \ref{sec:concentration}. We use these bounds to
prove Theorem \ref{thm:intro-product} on products of low-variance
variables in Section \ref{sec:products} and to analyze the \cite{GMRTV}
generator in Section \ref{sec:gmr}.

\section{Inequalities for symmetric polynomials}
\label{sec:inequalities}

\begin{proof}[Proof of Theorem~\ref{thm:boundUsingS1E2}]
It will be convenient to use
$$E_2(a) = \sum_{i \in [n]} a_i^2.$$
By Newton's identity, 
$E_2 = S_1^2 - 2S_2$ so for all $a \in \R^n$,
$$S^2_1(a) + E_2(a) \leq 2 (S^2_1(a) + |S_2(a)|).$$
It therefore suffices to prove that for all $a \in \R^n$ and $k \in [n]$,
\[ S^2_k(a) \leq \frac{(16 e^2 (S^2_1(a) + E_2(a)))^{k} }{k^{k}}.\]
We prove this by induction.
For $k \in \{1,2\}$, it indeed holds.
Let $k > 2$.
Our goal will be upper bounding the maximum of
the projectively defined\footnote{That is,
for every $a \neq 0$ in $\R^n$ and $c \neq 0$ in $\R$,
we have $\phi_k(ca) = \phi_k(a)$.} function
$$\phi_k(a) = 
\frac{S^2_k(a)}{(S^2_1(a) + E_2(a))^{k} }$$
under the constraint that $S_1(a)$ is fixed. 
Since $\phi_k$ is projectively defined, 
its supremum is attained in the (compact) unit sphere,
and is therefore a maximum.
Choose $a \neq 0$ to be a point that achieves the maximum
of $\phi_k$.
We assume, without loss of generality, that $S_1(a)$
is non-negative (if $S_1(a)<0$, consider $-a$ instead of $a$).
There are two cases to consider:

The first case is that for all $i \in [n]$, 
\begin{equation}
\label{eq:case1}
a_i \leq \frac{2k^{1/2} (S_1^2(a) + E_2(a))^{1/2}}{n} .
\end{equation}
In this case we do not need the induction hypothesis and can in fact replace each
$a_i$ by its absolute value. Let $P \subseteq [n]$ be the set of $i\in
[n]$ so that $a_i \geq 0$. Then by Equation \eqref{eq:case1}, 
\[ \sum_{i  \in P} |a_i|  \leq 2k^{1/2} (S_1^2(a) + E_2(a))^{1/2}.\]
Note that
\[ S_1(a) = \sum_{i \in P}|a_i| - \sum_{i \not\in P}|a_i| \geq 0.\]
Hence 
\[ \sum_{i \not\in P}|a_i| \leq \sum_{i \in P}|a_i| \leq 2k^{1/2} (S_1^2(a) + E_2(a))^{1/2}.\]
Overall we have
\[ \sum_{i  \in [n]} |a_i|  \leq 4k^{1/2} (S_1^2(a) + E_2(a))^{1/2}.\]
We then bound
\begin{align*}
|S_k(a_1,\ldots,a_n)| & \leq S_k(|a_1|,\ldots,|a_n|)\\
& \leq \left(\frac{e}{k}\right)^k\left(\sum_{i \in [n]}|a_i|\right)^k \ \ \text{By the Maclaurin identities}\\
& \leq \left(\frac{4e}{\sqrt{k}}\right)^k(S_1^2(a) + E_2(a))^{k/2}.
\end{align*}


The second case is that there exists $i_0 \in [n]$ so that
\begin{align}
\label{eqn:aiIsgood}
a_{i_0} > \frac{2k^{1/2} (S^2_1(a) + E_2(a))^{1/2} }{n} .
\end{align}
In this case we use induction and Lagrange multipliers.
For simplicity of notation, for a function $F$ on $\R^n$ denote 
$$F(-i) = F(a_1,a_2,\ldots,a_{i-1},a_{i+1},\ldots,a_n)$$
for $i \in [n]$.
So, for every $\delta \in \R^n$ so that $\sum_i \delta_i = 0$
we have $\phi_k(a+\delta) \leq \phi_k(a)$.
Hence\footnote{Here and below,
$O(\delta^2)$ means of absolute value at most 
$C \cdot \|\delta\|_\infty$
for $C = C(n,k) \geq 0$.}, for all $\delta$ so that $\sum_i \delta_i = 0$,
\begin{align*}
 \phi_k(a) 
& \geq \frac{S^2_k(a+\delta)}{(S^2_1(a+\delta) + E_2(a+\delta))^{k} } \\
& \geq \frac{(S_k(a) + \sum_i \delta_i S_{k-1}(-i) + O(\delta^2))^2}{
(S^2_1(a) + E_2(a) + 2 \sum_i a_i \delta_i + O(\delta^2))^{k} } \\ 
& \geq \frac{S^2_k(a) + 2 S_k(a) \sum_i \delta_i S_{k-1}(-i) + O(\delta^2)}{
(S^2_1(a) + E_2(a))^k + 2 k (S^2_1(a) + E_2(a))^{k-1} \sum_i a_i \delta_i + O(\delta^2)} .
\end{align*}
Hence, for all $\delta$ close enough to zero so that $\sum_i \delta_i = 0$,
\begin{align*} 
\frac{S^2_k(a)}{(S^2_1(a) + E_2(a))^{k} }
& \geq \frac{S^2_k(a) + 2 S_k(a) \sum_i \delta_i S_{k-1}(-i) + O(\delta^2)}{
(S^2_1(a) + E_2(a))^k + 2 k (S^2_1(a) + E_2(a))^{k-1} \sum_i a_i \delta_i + O(\delta^2)} ,
\end{align*}
or
\begin{align}
\label{eqn:resp0} 
  \sum_i \delta_i \left( a_i S_k(a)  k  -  
  (S^2_1(a) + E_2(a))  S_{k-1}(-i) \right)
\geq 0 .
\end{align}
For the above inequality to hold for all such $\delta$,
it must be that there is $\lambda$ so that for all $i \in [n]$,
\begin{align*}
a_i S_k(a)  k  -  
  (S^2_1(a) + E_2(a))  S_{k-1}(-i) = \lambda .
\end{align*}
To see why this is true, set $\lambda_i = a_i S_k(a)  k  -  
  (S^2_1(a) + E_2(a))  S_{k-1}(-i)$ .
We now have $\lambda_1,\ldots,\lambda_n$ so that  
\begin{equation} 
\label{eqn:resp1}
\sum_i \lambda_i\delta_i \geq 0
\end{equation}
for every $\delta_1,\ldots,\delta_n$ of sufficiently small norm where $\sum_i \delta_i =0$.
We claim that this implies that in fact $\lambda_i = \lambda$
for every $i$.
To see this, assume for contradiction that $\lambda_1 \neq \lambda_2$
and $|\lambda_1| > |\lambda_2|$. Set 
\[\delta_1 = -\mu\lambda_1, \ \delta_2 = \mu\lambda_1, \ \delta_3
= \delta_4 = \ldots = \delta_n =0\]
for $\mu > 0$ sufficiently small. It follows that
$\sum_i \delta_i =0$ and $\sum_i \lambda_i\delta_i 
= \mu ( \lambda_1 \lambda_2 - \lambda_1^2) < 0$
so Equation \eqref{eqn:resp1} is violated.

Sum over $i$ to get
\begin{align*}
\lambda n 
& = S_1(a) S_k(a)  k  - (S^2_1(a) + E_2(a)) (n-(k-1)) S_{k-1}(a) .
\end{align*}
Thus, for all $i \in [n]$,
\begin{align*}
a_i S_k(a)  k  & -  
  (S^2_1(a) + E_2(a))  S_{k-1}(-i) \\ 
&  = \frac{1}{n}\left(
S_1(a) S_k(a)  k  - (S^2_1(a) + E_2(a)) (n-(k-1)) S_{k-1}(a)\right) ,
\end{align*}
or
\begin{align*}
 S_k(a) &  k \left(a_i - \frac{S_1(a)}{n} \right)  \\ 
&  = 
  (S^2_1(a) + E_2(a))  (S_{k-1}(-i) - S_{k-1}(a)) +  \frac{ (k-1)}{n}  (S^2_1(a) + E_2(a)) S_{k-1}(a)) .
\end{align*}
This specifically holds for $i_0$, so using \eqref{eqn:aiIsgood}
we have
\begin{align*}
& \left| S_k(a)  k \frac{a_{i_0}}{2} \right| \\ 
& < \left| S_k(a)  k \left(a_{i_0} - \frac{S_1(a)}{n} \right) \right| \\ 
  &  \leq 
\left|  (S^2_1(a) + E_2(a))  a_{i_0}  S_{k-2}(-i_0) \right|
 + \left|  \frac{ (k-1)  (S^2_1(a) + E_2(a)) S_{k-1}(a)}{n} \right| ,
\end{align*}
or
\begin{align}
\label{eqn:Ski0} & \left| S_k(a)  \right| \\ 
\notag  &  \leq 
\left| \frac{ 2 (S^2_1(a) + E_2(a))   S_{k-2}(-i_0)}{k} \right|
 + \left| \frac{ 2 (k-1)  (S^2_1(a) + E_2(a)) S_{k-1}(a)}{n k a_{i_0} } \right| \\
  \notag &  <
\left| \frac{ 2 (S^2_1(a) + E_2(a))   S_{k-2}(-i_0)}{k} \right|
 + \left| \frac{   (S^2_1(a) + E_2(a))^{1/2} S_{k-1}(a)}{k^{1/2}  } \right| .
\end{align}
To apply induction we need to bound
$S_1^2(-i_0)+E_2(-i_0)$ from above. Since
\begin{align*}
S_1^2(a) + E_2(a)
- S_1^2(-i_0)-E_2(-i_0) 
& = a_{i_0}^2 + 2 a_{i_0} S_1(-i_0)
+ a_{i_0}^2 \\
& = 2 a_{i_0} S_1(a) \geq 0 .
\end{align*}
we have the bound
\begin{align*}
S_1^2(-i_0)+E_2(-i_0) \leq S_1^2(a) + E_2(a).
\end{align*}
Finally, by induction and \eqref{eqn:Ski0},
\begin{align*}
 \left| S_k(a)  \right| 
 &  \leq
 \frac{ 2 (S^2_1(a) + E_2(a))}{k}  \frac{(16 e^2 (S^2_1(-i_0) + E_2(-i_0)))^{(k-2)/2} }{(k-2)^{(k-2)/2}}
 \\ & + \frac{   (S^2_1(a) + E_2(a))^{1/2} }{k^{1/2}  } 
 \frac{(16 e^2 (S^2_1(a) + E_2(a)))^{(k-1)/2} }{(k-1)^{(k-1)/2}} \\
 &  \leq \frac{ (16e^2 (S^2_1(a) + E_2(a)))^{k/2} }{k^{k/2}} \left(
 \frac{ 2 }{16 e^2 \left( 1- \frac{2}{k} \right)^{(k-2)/2}} + 
 \frac{1 }{4e \left( 1- \frac{1}{k} \right)^{(k-1)/2}} \right) \\
 &  < \frac{ (16e^2 (S^2_1(a) + E_2(a)))^{k/2} }{k^{k/2}} .
\end{align*}

\end{proof}

\begin{proof}[Proof of Theorem~\ref{thm:GenBound}]
The proof is by reduction to Theorem~\ref{thm:boundUsingS1E2}.
Assume $a_1,\ldots,a_m$ are nonzero and $a_{m+1},\ldots,a_n$ are zero.
Denote $a' = (a_1,\ldots,a_m)$ and notice
that for all\footnote{For $k > m$ we have $S_k(a) =0$ so there is nothing to prove.}
$k \in [n]$,
$$S_k(a) = S_k(a').$$
Write
$$p(\xi) = \prod_{i \in [m]} (\xi a_i+ 1) =
\sum_{k = 0}^m \xi^k S_k(a) .$$
Derive $k$ times to get
\begin{align*}
p^{(k)}(\xi)   =
S_k(a) k! \left( {m \choose k} \frac{S_m(a)}{S_k(a)} \xi^{m-k} \right.
 + {m-1 \choose k} & \left. \frac{S_{m-1}(a)}{S_k(a)} \xi^{m-k-1} + \ldots \right. \\
& \left.  \ldots + 
 {k+1 \choose k} \frac{S_{k+1}(a)}{S_{k}(a)} \xi+ 1\right).
\end{align*}
Since $p$ has $m$ real roots,  
$p^{(k)}$ has $m-k$ real roots.
Since $p^{(k)}(0) \neq 0$,
there is $b \in \R^{m-k}$ so that
$$p^{(k)}(\xi) = S_k(a) k! \prod_{i \in [m-k]} (\xi b_i + 1).$$
For all $h \in [m-k]$,
$$S_h(b) = {k+h \choose k} \frac{S_{k+h}(a)}{S_{k}(a)}.$$
By assumption,
$$|S_1(b)| \leq C 
\ \ \text{and} \ \
|S_2(b)| \leq C^2.$$
Theorem~\ref{thm:boundUsingS1E2} implies
$$|S_h(b)| = \left| {k+ h \choose k} \frac{S_{k+h}(a)}{S_{k}(a)} \right| \leq \frac{(6 e C)^h }{h^{h/2}} .$$
\end{proof}

\section{Tail bounds under limited independence}
\label{sec:concentration}

In this section we work with the following setup: Let $X
=(X_1,\ldots,X_n)$ be a vector of real valued random variables so that
$\E[X_i] = 0$ for all $i \in [n]$.
Denote $\sigma_i^2 = \Var[X_i]$ and   
\[\sigma^2 = \sum_{i \in [n]} \sigma_i^2.\]  
The goal is proving a tail bound on the behaviour
of the symmetric functions under limited independence.

We start by obtaining tail estimates,
under full independence.
Let $\cal{U}$ denote the distribution over
$X=(X_1,\ldots,X_n)$ where $X_1,\ldots,X_n$ are independent.

\begin{lemma}
\label{lem:exp-kwise}
$\E_{X \in \cal{U}}[S^2_\ell(X)] \leq \frac{\sigma^{2\ell}}{\ell!}$.
\end{lemma}
\begin{proof}
Since the expectation of $X_i$ is zero for all $i \in [n]$,
\begin{align*}
\E[S^2_\ell(X)]
& =
\sum_{T,T' \subset [n]:|T|=|T'| = \ell}
\E\left[\prod_{t \in T} X_{t} \prod_{t' \in T'} X_{t'}\right]\\
& = \sum_{T \subset [n]:|T| = \ell}
\E\left[\prod_{t \in T} X^2_{t} \right]
 = \sum_{T \subset [n]:|T| = \ell} \prod_{t \in T} \sigma_t^2 \\
& \leq \frac{1}{\ell!}\left( \sum_{i \in [n]} \sigma_i^2 \right)^\ell
 = \frac{\sigma^{2\ell}}{\ell!} .
\end{align*}
\end{proof}

\begin{corollary}
\label{cor:k-wise}
For $t >0$ and $\ell \in [n]$,
by Markov's inequality,
\begin{align}
\label{eq:tail1}
\Pr_{X \in \cal{U}} \left[|S_\ell(X)| \geq 
\left( \frac{e^{1/2} t \sigma }{ \ell^{1/2} } \right)^{\ell} \geq \frac{(t\sigma)^\ell}{\sqrt{\ell!}}\right]
\leq \frac{1}{t^{2\ell}}.
\end{align}
If $2 e^{1/2} t\sigma \leq k^{1/2}$ then
by the union bound
\begin{align}
\label{eq:tail22}
\Pr_{X \in \cal{U}} \left[\sum_{\ell = k}^n |S_\ell(X)| \geq 2
\left( \frac{e^{1/2} t \sigma }{ k^{1/2} } \right)^{k} \right]
\leq \frac{1}{t^{2k} - t^{2(k-1)}}.
\end{align}
\end{corollary}

We now consider limited independence.
\begin{lemma}
\label{lem:main}
Let $\cal{D}$ denote a distribution over $X = (X_1,\ldots,X_n)$ where $X_1,\ldots,X_n$ are
$(2k+2)$-wise independent. Let $t \geq 1$. Except with $\cal{D}$-probability
at most $2t^{-2k}$, 
the following bounds hold for all $\ell \in \{k,\ldots,n\}$:
\begin{align}
\label{eq:kwise-tail1}
|S_\ell(X)| \leq (6 e t\sigma)^{\ell}  
\left( \frac{k}{\ell} \right)^{\ell/2}  .
\end{align}
\end{lemma}
\begin{proof}
In the following the underlying probability distribution over $X$ is $\cal{D}$.
By Lemma \ref{lem:exp-kwise}, for $i \in \{k,k+1\}$,
\begin{align*}
\E[S^2_{i}(X)] & \leq \frac{\sigma^{2i}}{i!} .
\end{align*}
Hence by Markov's inequality,
\begin{align*}
\Pr\left[|S_{i}(X)| \geq
  \frac{(t\sigma)^i}{\sqrt{i!}}\right] & \leq t^{-2 i}.  
\end{align*}
From now on, condition on the event that 
\begin{align}
\label{eq:condition1}
|S_{k}(X)| \leq \frac{(t\sigma)^k}{\sqrt{k!}} \ \text{and}
\ |S_{k+1}(X)| \leq \frac{(t\sigma)^{k+1}}{\sqrt{(k+1)!}} ,
\end{align}
which occurs with probability at least $1-2t^{-2k}$. 
Fix $x = (x_1,\ldots, x_n)$ such that Equation \eqref{eq:condition1}
holds.

We claim that there must exist $k_0 \in \{0,\ldots,k-1\}$ for which the
following bounds hold:
\begin{align}
\label{eq:decay1}
|S_{k_0}(x)| & \geq \frac{(t\sigma)^{k_0}}{\sqrt{k_0!}} , \\
\label{eq:decay2}
|S_{k_0+1}(x)| & \leq \frac{(t\sigma)^{k_0+1}}{\sqrt{(k_0+1)!}} ,\\
\label{eq:decay3}
|S_{k_0 + 2}(x)| & \leq \frac{(t\sigma)^{k_0 + 2}}{\sqrt{(k_0 + 2)!}} .
\end{align}
To see this, mark point $j \in \{0,\ldots,k+1\}$ as {\em high} if 
\begin{align*}
|S_j(x)| & \geq \frac{(t\sigma)^j}{\sqrt{j!}}
\end{align*}
and {\em low} if 
\begin{align*}
|S_j(x)| & \leq \frac{(t\sigma)^j}{\sqrt{j!}} .
\end{align*}
A point is marked both high and low if equality holds. Observe that
$0$ is marked high (and low) since $S_0(x) =1$ and $k$ and $k +1$
are marked low by Equation \eqref{eq:condition1}. This implies the
existence of a triple $k_0,k_0+1,k_0+2$ where the first point is
high and the next two are low.

Let $\gamma > 0$ be the smallest number so that the following inequalities hold:
\begin{align}
\label{eq:def-gamma1}
|S_{k_0 +1}(x)| & \leq |S_{k_0}(x)|
\frac{\gamma}{\sqrt{k_0 +1}}, \\
\label{eq:def-gamma2}
|S_{k_0 +2}(x)| & \leq |S_{k_0}(x)|
\frac{\gamma^2}{\sqrt{(k_0 +1)(k_0 +2)}} .
\end{align}
By definition, 
one of Equations \eqref{eq:def-gamma1} and \eqref{eq:def-gamma2} holds with
equality so
\begin{align*}
|S_{k_0}(x)| = \max\left\{\frac{|S_{k_0
    +1}(x)|\sqrt{k_0 +1}}{\gamma}, \frac{|S_{k_0
    +2}(x)|\sqrt{(k_0 +1)(k_0 +2)}}{\gamma^2}\right\} .
\end{align*}
Observe further that $\gamma \leq t\sigma$ by Equations \eqref{eq:decay1},
\eqref{eq:decay2} and \eqref{eq:decay3}.
Combining this with the bounds in Equations
\eqref{eq:decay2} and \eqref{eq:decay3}
\begin{align}
\label{eq:bound-sk}
|S_{k_0}(x)| \leq \max\left\{\frac{(t\sigma)^{k_0+1}}{\gamma\sqrt{k_0!}}, \frac{(t\sigma)^{k_0+2}}{\gamma^2\sqrt{k_0!}} \right\} = \frac{(t\sigma)^{k_0+2}}{\gamma^2\sqrt{k_0!}} .
\end{align}

Equations \eqref{eq:def-gamma1}
and \eqref{eq:def-gamma2} let us apply Theorem \ref{thm:GenBound}
with $C = \gamma \sqrt{k_0+1}$ and $h \geq 3$ 
to get
\begin{align*}
\left|\frac{S_{k_0+h}(x)}{S_{k_0}(x)}\right| \leq (6 e \gamma)^h
\frac{(k_0 +1)^{h/2}}{h^{h/2}{ k_0 + h \choose k_0}}.
\end{align*}
Bounding $|S_{k_0}|$ by Equation \eqref{eq:bound-sk}, we get
\begin{align*}
|S_{k_0+h}(x)| & \leq (6 e \gamma)^h
\frac{(k_0 +1)^{h/2}}{h^{h/2}{ k_0 + h \choose k_0}}
\frac{(t\sigma)^{k_0+2}}{\gamma^2\sqrt{k_0!}}
 \leq (6 e t\sigma)^{k_0 + h}\frac{(k_0 + 1)^{h/2}}{h^{h/2}\sqrt{k_0!}{ k_0 + h \choose h}} .
\end{align*}
Since
\begin{align*}
{k_0 + h \choose h} \geq 
\max \left\{ \left(\frac{k_0 + h}{k_0}\right)^{k_0},
 \left(\frac{k_0 + h}{h}\right)^{h} \right\}
 \geq \frac{(k_0 + h)^{(k_0+h)/2}}{k_0^{k_0/2} h^{h/2}} ,
\end{align*}
we have
\begin{align*}
\frac{(k_0+1)^{h/2}}{h^{h/2}\sqrt{k_0!}{ k_0 + h \choose h}}
& \leq
\left( \frac{k_0+1}{h} \right)^{h/2} 
\frac{k_0^{k_0/2} h^{h/2}}{(k_0 + h)^{(k_0+h)/2}} 
 \leq
\left( \frac{k_0+1}{k_0 + h} \right)^{(k_0+h)/2} .
\end{align*}
Therefore, denoting $\ell = k_0 + h$, since $k_0 +1 \leq k$,
\begin{align*}
|S_{\ell}(x)| & \leq (6 e t\sigma)^{\ell}  
\left( \frac{k}{\ell} \right)^{\ell/2}  .
\end{align*}

\end{proof}

\begin{proof}[Proof of Theorem \ref{thm:main}]
As in Lemma~\ref{lem:main},
fix $x = (x_1,\ldots, x_n)$ such that Equation \eqref{eq:condition1} holds
(the random vector $X$ has this property with $\cal{D}$-probability
at least $1-2t^{-2k}$).
By the proof of lemma, since by assumption $6 e t\sigma < 1/2$, 
\begin{align}
\label{eqn:Endsum}
\sum_{\ell=k}^n|S_\ell(x)| \leq \frac{(t\sigma)^k}{k!} +
\frac{(t\sigma)^{k+1}}{\sqrt{(k+1)!}} + \sum_{\ell = k+2}^n (6 e t\sigma)^\ell 
\left( \frac{k}{\ell} \right)^{\ell/2}
\leq 2(6 e t\sigma)^k.
\end{align}
\end{proof}

\subsection{On the tightness of the tail bounds}
\label{sec:lower}

We conclude by showing that $(2k +2)$-wise independence is
insufficient to fool $|S_\ell|$ for $\ell > 2k +2$ in expectation. We
use a modification of a simple proof due to Noga Alon of the $\Omega(n^{k/2})$
lower bound on the support size of a $k$-wise independent
distribution on $\{-1,1\}^n$, which was communicated to us by Raghu Meka.

For this section, let $X_1,\ldots,X_n$ be 
so that each $X_i$ is uniform over
$\{-1,1\}$. 
Thus $\sigma^2 = \sum_i \Var[X_i] = n$.
By Lemma \ref{lem:exp-kwise}, we have
\begin{align}
\label{eq:k-wise-upper}
\E_{X \in \cal{U}}[|S_\ell(X)|] \leq \left(\E_{X \in \cal{U}}[S^2_\ell(X)]\right)^{1/2} 
\leq \frac{n^{\ell/2}}{\sqrt{\ell!}}.
\end{align}
In contrast we have the following:

\begin{lemma}
\label{lem:lower bound}
There is a $(2k+2)$-wise independent distribution 
on $X = (X_1,X_2,\ldots,X_n)$ 
in $\{-1,1\}^n$ such that for every $\ell \in [n]$,
\begin{align*}
\Pr_{X \in \cal{D}}\left[|S_\ell(X)| \geq {n \choose \ell}\right] \geq \frac{1}{3n^{k+1}} .
\end{align*}
Specifically,
\begin{align}
\label{eq:k-wise-lower2}
\E_{X \in \cal{D}}[|S_\ell(X)|] \geq \frac{{n \choose \ell}}{3n^{k+1}} .
\end{align}

\end{lemma}
\begin{proof}
Let $\cal{D}$ be a $(2k+2)$-wise independent distribution on $\{-1,1\}^n$ that is uniform
over a set $D$ of size $2(n+1)^{k+1} \leq 3n^{k+1}$. Such distributions are known to exist \cite{ABI}. 
Further, by translating the support by some fixed vector if needed, we
may assume that $(1,1,\ldots,1) \in D$. It is easy to see that every
such translate also induces a $(2k+2)$-wise independent distribution. 
The claim holds since
$S_\ell(1,\ldots,1) = {n \choose \ell}$.
\end{proof}

When e.g.\ $k = O(\log n)$, which is often the case of interest, 
for $2k+3 \leq \ell \leq n - (2k +3)$, the RHS of \eqref{eq:k-wise-lower2} is much larger than the bound guaranteed by Equation \eqref{eq:k-wise-upper}.
The tail bound provided by Lemma \ref{lem:main} 
can not therefore be extended to a satisfactory bound on the expectation.
Furthermore, applying Lemma \ref{lem:main} with 
\[ t = \frac{1}{6e}\sqrt{\frac{n}{\ell k}} \]
implies that for any $(2k +2)$-wise independent distribution, 
\begin{align*}
\Pr\left[|S_\ell(X)| \geq {n \choose \ell}\right] \leq
\Pr\left[|S_\ell(X)| \geq
  (6et\sqrt{n})^\ell\left(\frac{k}{\ell}\right)^{\ell/2} \right] \leq
2\left(\frac{36e^2 k\ell}{n}\right)^k . 
\end{align*}
When $k \ell = o(n)$, 
this is at most $O(n^{-k + o(1)})$.
Comparing this to the bound given in Lemma~\ref{lem:lower bound},
we see that the bound provided by Lemma~\ref{lem:main} is 
nearly tight.

\section{Limited independence fools products of bounded variables}
\label{sec:products}

In this section we work with the following setup. We have $n$ random
variables $X_1,\ldots, X_n$ each distributed in the interval
$[-1,1]$. Let $\mu_i$ and $\sigma_i^2$ denote the mean and variance of
$X_i$, and let $\sigma^2 = \sum_{i=1}^n \sigma_i^2$. We will typically
use $\U$ to denote the distribution where the $X_i$s are fully
independent, and $\D$ to denote distributions with limited independence.

\begin{theorem}
\label{thm:product}
There exist  constants $c, c' > 0$ such that under any $ck$-wise independent
distribution $\D$, 
\begin{align}
\label{eq:k-wise2}
\left|\E_\D[\prod_{i=1}^n X_i] - \prod_{i=1}^n\mu_i \right| \leq (c'\sigma)^k. 
\end{align}
\end{theorem}
\begin{proof}
Define $H \subset [n]$ to be the set of indices such that $|\mu_i| \leq
\sqrt{\sigma}$. Note that if $H \geq 2k$, then we are done since if $c
\geq 2$, then
\[
\abs{\E_\D[\prod_{i \in H} X_i]} = \prod_{i \in H}\abs{\mu_i} \leq \sqrt{\sigma}^{2k} \leq \sigma^k
\]
Further, since the variables are bounded in $[-1,1]$, we have 
\[ \abs{\prod_{i \in n}X_i} \leq \abs{\prod_{i \in H}X_i} \]
hence
\[ \abs{\E_\D[\prod_{i \in [n]} X_i]} \leq \abs{\E_\D[\prod_{i \in H}X_i]} \leq \sigma^k.\]
The same bound also holds under $\U$, hence 
\[ \abs{\E_\D[\prod_{i \in [n]} X_i] - \E_\U[\prod_{i \in [n]}X_i]} \leq 2\sigma^k.\]

So now assume that $|H| \leq 2k$. Let $T = H \setminus [n]$. 
Even after conditioning on the outcome of variables in $H$, the
resulting distribution on $T$ is $(c-2)k = c''k$-wise independent. Since the
product of variables in $H$ has absolute value at most $1$, it suffices
to show that for a $c''k$-wise independent distribution $\D$, 
\[ \left| \E_\D[\prod_{i \in T} X_i] - \E_\U[\prod_{i \in T}X_i]\right| \leq 2\sigma^k.\]
For ease of notation, we shall assume that $T = [m]$ for some $m \leq
n$. We may assume that $m >c''k$ else there is nothing to prove.

Let us write $X_i = \mu_i(1 + Z_i)$, so that $Z_i$ has mean $0$ and
variance $\sigma_i^2/\mu_i^2$. We write
\[ \prod_{i \in [m]}X_i = \prod_{i \in [m]}\mu_i(1 +Z_i) = \prod_{i \in
  [m]}\mu_i\left(\sum_{\ell \leq m}S_\ell(Z_1,\ldots,Z_m)\right) \].
Let us define the functions 
\begin{align*} 
P(Z) & = \prod_{i \in [m]}\mu_i(1 +Z_i) \\
P'(Z) & = \prod_{i \in [m]}\sum_{\ell \leq 4k}S_\ell(Z_1,\ldots,Z_m).
\end{align*}
We will prove the following claim.
\begin{claim}
\label{claim:tech}
For a $c''k$-wise independent distribution $\D$, 
\[ \abs{\E_\D[P(Z)] - \E_\D[P'(Z)]} \leq (c'\sigma)^k/2.\]
\end{claim}

We first show how to finish the proof of Theorem \ref{thm:product}
with this claim. We have
\begin{align*}
\abs{\E_\D[P(Z)] - \E_\U[P(Z)]} & \leq \abs{\E_\D[P(Z)] - \E_\D[P'(Z)]} 
+  \abs{\E_\U[P(Z)] - \E_\U[P'(Z)]}\\
& \phantom{=}  + \abs{\E_\D[P'(Z)] - \E_\U[P'(Z)]} 
\end{align*}
The first two are bounded by $(c'\sigma^k)/2$ by the claim, and the last is
$0$ since $c''k$-wise independence fools degree $4k$ polynomials for $c''
> 4$.

\begin{proof}[Proof of Claim \ref{claim:tech}]
Recall that the $X_i$s for $i \in [m]$ have expectation $\mu_i$ where $|\mu_i| \geq
\sqrt{\sigma}$. We let $X_i = \mu_i(1 + Z_i)$, where $Z_i$ has mean $0$ and
variance $\bar{\sigma}_i^2$ where
\[ \bar{\sigma}_i^2 = \frac{\sigma_i^2}{\mu_i^2} \leq \frac{\sigma_i^2}{\sigma}.\]
Hence the total variance of the $Z_i$s can be bounded by
\[ \sig^2 \leq \sum_{i \in T} \frac{\sigma_i^2}{\sigma} \leq \sigma. \]

Writing $Z = (Z_1,\ldots,Z_m)$ we have
\[ P(Z) -P'(Z) = \sum_{\ell=4k +1}^m |S_\ell(Z)|.\]
Let $G$ denote the event that $|P(Z) - P'(Z)| \leq 2(6 e \sqrt{\sig})^{4k}$.
Letting $t = 1/\sqrt{\sig}$ and applying Theorem \ref{thm:main}, for
$c'' > 8k +2$
\begin{align}
\label{eq:main}
\E_\D[\ind(\neg G)]  \leq 2t^{-8k} = 2\sig^{4k}.
\end{align}

Since $\E[Z_i] =0$ for all $i$ it follows that under $c''k$-wise independence,
\begin{align}
\label{eq:pkZ}
\E[P_k(Z_1,\ldots,Z_n)^2 ] & \leq
\sum_{i=0}^{4k}\E[S_i(Z_1,\ldots,Z_n)^2] \leq \sum_{i=0}^{4k}\frac{\sig^{2i}}{i!} \leq 2.
\end{align}

We now write
\begin{align*} 
\E[P(Z) - P'(Z)] = \E[(P(Z) - P'(Z))\mathbbm{1}(G)] + \E[(P(Z) - P'(Z))\mathbbm{1}(\neg G)] .
\end{align*}
Equation \eqref{eq:tail-sum} implies
\begin{align*}
|\E[(P(Z) - P'(Z))\mathbbm{1}(G)]| \leq 2(6e\sqrt{\sig})^{4k} .
\end{align*}
For the second term,
\begin{align*}
|\E[(P(Z) - P'(Z))\mathbbm{1}(\neg G)]|  \leq |\E[P(Z)\mathbbm{1}(\neg G)]| + |\E[P'(Z)\mathbbm{1}(\neg G)]|
\end{align*}
Note that $0 \leq P(Z) \leq 1$. Also note that $\E[P_k(Z)^2] \leq 2$ by Equation \eqref{eq:pkZ}. 
So we can bound the RHS using Holder's inequalities by
\begin{align*}
|\E[\mathbbm{1}(\neg G)]| + |\E[P_k(A)^2]^{1/2} \cdot \E[\mathbbm{1}(\neg
  G)]^{1/2}| \leq \sig^{4k} + \sqrt{2}\sig^{2k} \leq 2\sig^{2k}.
\end{align*}
Hence overall we have
\[ \E[P(Z) - P'(Z)] = 2(6e\sqrt{\sig})^{4k} + 2\sig^{2k} \leq (c'\sigma)^k/2\]
\end{proof}
\end{proof}

\section{Analyzing the \cite{GMRTV} generator}
\label{sec:gmr}

Gopalan et al.~\cite{GMRTV} proposed and analyzed a $\prg$
for combinatorial rectangles, which we denote by $\gmr$.
In this section, we provide a different analysis of their construction,
which is based on our results concerning the symmetric
polynomials. Our analysis is simpler and follows the intuition  that
products of low variance events are easy to fool using limited
independence. It also improves one their seedlength in the dependence
on $n, \delta$ (see the discussion following Theorem \ref{thm:gmr}).

Let $\U$ denote the uniform distribution on $[m]^n$,
and let $\D$ be a distribution on $[m]^n$.
For $x \in [m]^n$ and $K \subseteq [n]$, let $x_K = (x_i : i \in K)$. 
We sometimes abuse notation
and write $x_K$ instead of the probability distribution of $x_K$.
We denote by $d_{TV}$ the total variation distance.

\begin{definition}
A distribution $\D$ on $[m]^n$ is $(k,\eps)$-wise independent
if for every $K \subseteq [n]$ of size $k$, and $x \in \D, y \in \U$, 
we have $d_{TV}(x_K,y_K) \leq \eps$.
\end{definition}

Such distributions can be generated using
seed length $O(\log\log(n) + k\log(m) + \log(1/\eps))$ when $m$ is a
power of $2$ using standard constructions \cite{NaorN93}. 
We can also assume that every co-ordinate is uniformly random in
$[m]$. See the appendix for details.

(by adding the string
$(a,a,\ldots,a)$ modulo $m$, where $a \in [m]$ is uniformly random). 

Being $(k,\eps)$-wise independent is equivalent to saying that 
for every $K \subseteq [n]$ of size $k$ and every $f:[m]^k \rgta \zo$, 
\[ \left|\E_{x \in \D}[f(x_K)] - \E_{y \in \U}[f(y_K)]\right| \leq \eps. \]
The following more general property holds.
Let $P$ be a real linear combination of combinatorial rectangles,
\[P = \sum_{S \subseteq [n]} c_S f_S ,\]
where $f_S(x) = \prod_{i \in S} f_i(x_i)$.
Let $\nm(P) = \sum_S |c_S|$. 
The degree of $P$ is the maximum size of $S$ for which $c_S \neq 0$.
It follows that if $\D$ is $(k,\eps)$-wise
independent and $P$ has degree at most $k$ then
\[ \left|\E_{x \in \D}[P(x)] - \E_{x \in \U}[P(x)]\right| \leq \nm(P)\eps. \]

\subsection{The generator}
\label{sec:gmr-def}

We use an alternate view of $\gmr$ as a collection of hash functions $g:[n] \rgta [m]$.
The generator $\gmr$ is based on iterative applications of an {\em alphabet increasing} step.
The first alphabet $m_0$ is chosen to be large enough, 
and at each step $t>1$ the size of the alphabet $m_t$ is squared $m_t = m_{t-1}^2$.
There is a constant $C > 0$ so that the following holds.
Denote by $\delta$ the error parameter of the generator.
Let $T \leq C \log \log(m)$ be the first integer so that $m_T \geq m$.
Let $\delta' = \delta/T$.

\begin{enumerate}
\item {\bf Base Case: } 
Let $m_0 \geq C \log(1/\delta)$ be a power of $2$. 
Sample $g_0: [n] \rgta [m_0]$ using a $(k_0,\eps_0)$-wise
independent distribution on $[m_0]^n$ with   
\begin{align} 
\label{eq:g0}
k_0 = C \log(1/\delta'), \  \eps_0 = \delta' \cdot m_0^{-C k_0} .
\end{align}
This requires seed length $O(\log\log(n) + \log(\log\log(m)/\delta)\log\log(\log \log(m)/\delta))$.

\item {\bf Squaring the alphabet: }
Pick $g'_t: [m_{t-1}] \times [n] \rgta [m_t]$ using a $(k_t,\eps_t)$-wise independent distribution over $[m_t]^{m_{t-1} \times n}$ with
\begin{align*} 
k_t = \max \left\{ C \frac{\log(1/\delta')}{\log(m_t)}  , 2 \right\} , \
\eps_t  \leq m_t^{-C k_t} . 
\end{align*}
Define a hash function $g_t:[n] \rgta [m_t]$ as 
$$g_t(i) = g'_t(g_{t-1}(i),i).$$
This requires seed length $O(\log\log(n) + \log(m_t) + \log(\log \log(m)/\delta))$. 
\end{enumerate}

\eat{
The main property of the generator that 
we shall prove below is: 

\begin{theorem}
Let $\gmr$ be the family of hash functions from $[n]$ to $[m]$
with error parameter $\delta > 0$. 
The seed length is at most
\[ O((\log\log(n) + \log(m/\delta)) \log\log (m/\delta)).\]
Then, for every $S_1, \ldots,S_n \subseteq [m]$,
\[ \left|\Pr_{g \in \gmr}[\intr{i \in [n]}g(i) \in S_i] - \Pr_{h \in
  \uh}[\intr{i \in [n]}h(i) \in S_i]\right| \leq \delta. \]
\end{theorem}

Combining this with Theorem \ref{thm:prg-to-hash}, we get the
following corollary for min-wise independent hashing.

\begin{corollary}
For every $\ell$, there is a family of hash functions
that is approximately $\ell$-minima-wise independent with error $\eps$ and 
seed length at most 
$$O((\log\log(n) + \log(m^\ell/\eps))(\log\log(m^\ell/\eps))).$$
\end{corollary}
}

\subsection{Analyzing the generator}

We first analyze the base case using the inclusion-exclusion approach
of \cite{EGLNV}. We need to extend their analysis to the setting where
the co-ordinates are only approximately $k$-wise independent. 

\begin{lemma}
\label{lem:base}
Let $\D$ be a $(k,\eps)$-wise independent distribution 
on $[m]^n$ with $k$ odd.
Then,
\[ \left|\Pr_{g \in \D}[\intr{i \in [n]}g(i) \in S_i] - \Pr_{h \in
  \uh}[\intr{i \in [n]}h(i) \in S_i]\right| \leq \eps  m^k + \exp(-\Omega(k)) . \]
\end{lemma}

\begin{proof}
Let $p_i = |S_i|/m$, and $q_i = 1 - p_i$. Observe that 
\begin{align}
\label{eq:sum_q}
\Pr_{h}[\intr{i \in [n]}h(i) \in S_i] = \prod_{i=1}^np_i =
\prod_{i=1}^n(1 - q_i) \leq \exp \left(-\sum_{i=1}^nq_i \right).
\end{align}
We consider two cases based on $\sum_i q_i$.

{\bf Case 1: } When $\sum_i q_i \leq k/(2e)$. Since every
non-zero $q_i$ is at least $1/m$, there can be at most
$mk/(2e)$ indices $i$ so that $q_i > 0$. For $i$ so that $q_i
=0$, we have $S_i = [m]$, so we can drop such indices and assume $n \leq mk/(2e)$.
By Bonferroni inequality, since $k$ is odd,
\begin{align*}
& 
\left| \Pr_{g}[\intr{i \in [n]}g(i) \in S_i]  - \sum_{j=0}^{k-1} (-1)^j 
\sum_{J \subseteq [n]: |J| = j}\Pr_{g}[\intr{i \in J}g(i) \not\in S_i] \right|
\\ & \qquad \qquad \leq 
 \sum_{J \subseteq [n]: |J| = k} \Pr_{g}[\intr{i \in J}g(i) \not\in S_i] .
\end{align*}
A similar bound holds for $h$.
The $(k,\eps)$-wise independence thus implies
\begin{align*}
& \left|
 \Pr_{g}[\intr{i \in [n]}g(i) \in S_i]  - \Pr_{h}[\intr{i \in [n]}h(i) \in S_i]  \right| 
\\ & \qquad \qquad \leq 
\eps (en/k)^k + 2\sum_{J \subseteq [n]: |J| = k}\Pr_{h}[\intr{i \in J}h(i) \not\in S_i]  .
\end{align*}
The second term is twice $S_k(q_1,\ldots,q_n)$, which we can bound by
Maclaurin's identity as
\begin{align*} 
S_k(q_1,\ldots,q_n) \leq
\left(e/k\right)^k\left(\sum_{i=1}^nq_i\right)^k \leq 2^{-k} .
\end{align*}
Finally, since $n \leq mk/(2e)$, 
\[
\left|
 \Pr_{g}[\intr{i \in [n]}g(i) \in S_i]  - \Pr_{h}[\intr{i \in [n]}h(i) \in S_i]  - 
  S_i]
\right|  \leq \eps(m/2)^k + 2^{-k +1}.
\]

{\bf Case 2: } When $\sum_{i}q_i > k/2e$.
Once again, we drop indices $i$ so that $q_i =0$.
Consider the largest $n'$ such that 
\[ k/2e - 1 \leq \sum_{i =1}^{n'}q_i \leq k/2e.\]
Repeating the argument from Case 1 for this $n'$,
\[
\left|
 \Pr_{g}[\intr{i \in [n']}g(i) \in S_i]  - \Pr_{h}[\intr{i \in [n']}h(i) \in S_i]  \right|  \leq \eps(m/2)^k + 2^{-k+1}.
\]
Similarly to Equation~\eqref{eq:sum_q},
\[ \Pr_{h}[\intr{i \in [n']}h(i) \in S_i] \leq e^{-k/2e +1}. \]
Since
\[ \Pr_{g}[\intr{i \in [n]}g(i) \in S_i]  \leq 
\Pr_{g}[\intr{i \in[n']}g(i) \in S_i] ,\]
we have
\begin{align*}
\left|\Pr_{g}[\intr{i \in [n]}g(i) \in S_i]  \leq \Pr_{h}[\intr{i \in
    [n]}h(i) \in S_i]\right|  \leq  \eps(m/2)^k + \exp(-\Omega(k)).
\end{align*}
\end{proof}

To analyze the iterative steps, we use the
following lemma:

\begin{lemma}
\label{lem:induct}
There is $C > 0$ so that the following holds for $\delta > 0$ small enough.
Assume 
\begin{align*}  
k > 1, \ 
\ell \geq \log(1/\delta) , \ \ell \geq k , \ 
\ell^{-k} \leq \delta^{C} , \ 
\eps  \leq (m\ell)^{-C k} . 
\end{align*}
Let $\D$ be a $(k,\eps)$-wise independent distribution on $g' : [\ell] \times [n] \to [m]$ 
so that for every $(a,i) \in [\ell] \times [n]$
the distribution of $g'(a,i)$ is uniform on $[m]$.
Then,
\[ \left|\Pr_{g' \in \D, x \in [\ell]^n}[\intr{i \in [n]} g'(x_i,i) \in S_i] - \Pr_{h \in
  [m]^n}[\intr{i \in [n]}h(i) \in S_i]\right| \leq \delta. \]
\end{lemma}

\begin{proof}
Given $g',x$, let $g : [n] \to [m]$ be defined by $g(i) = g'(x_i,i)$.
We can similarly pick $h$ in two steps: 
pick $h':[\ell]\times [n] \rgta [m]$ uniformly at random, pick $x \in [\ell]^n$ 
independently and uniformly at random, and then let $h(i) = h'(x_i,i)$. 

For every $i \in [n]$,
since each $x_i$ is uniform over $[\ell]$, 
for every fixed $g'$, we have
\begin{align*}
\Pr_x[g(i) \in S_i] = \frac{1}{\ell}\sum_{a=1}^\ell
\mathbbm{1}(g'(a,i) \in S_i) .
\end{align*}
So, for every fixed $g'$,
\begin{align}
\label{eq:ind_x}
\Pr_x[\intr{i \in n}g(i)\in S_i] = \prod_{i=1}^n \Pr_x[g(i) \in S_i] = \prod_{i=1}^n\frac{1}{\ell}\sum_{a=1}^\ell \mathbbm{1}(g'(a,i) \in S_i) .
\end{align}
A similar equation holds for $h$.

Let $p_i =|S_i|/m$ and $q_i = 1- p_i$. 
Partition $[n]$ into a head
$H  = \{i : p_i < \ell^{-0.1}\}$ and a tail $T = \{i : p_i \geq \ell^{-0.1}\}$. 
Standard arguments (see e.g.~\cite[Theorem 4.1]{GMRTV})
imply that if $(k,\eps)$-wise independence fools both $\intr{i \in H}g(i) \in S_i$
and $\intr{i \in T}g(i) \in S_i$ with error $\delta$ then
$(O(k),\eps^{O(1)})$-wise independence fools their intersection with
error $O(\delta)$. 
So it suffices to consider each of them separately.

\paragraph{Fooling the Head:} If $|H| \leq k$, 
\begin{align*}
\Pr_x[\intr{i \in H}g(i) \in S_i] = \prod_{i \in H}\frac{1}{\ell}\sum_{a=1}^\ell \mathbbm{1}(g'(a,i) \in
S_i)
\end{align*}
is a degree $k$ polynomial with $\nm$-norm bounded by $1$. Hence,
\begin{align}
\label{eq:head}
\left| \E_{g'}[\Pr_x[\intr{i \in H} g(i) \in S_i]] -
\E_{h'}[\Pr_x[\intr{i \in H} h(i) \in S_i]] \right| \leq \eps \leq \delta.
\end{align}

If $|H| \geq k$, we show that the probabilities are small
which means that they are close.
Indeed, let $H'$ be the first $k$ indices in $H$.
First,
\begin{align*}
\Pr_{h}[\intr{i \in H'} h(i) \in S_i] = 
\prod_{i \in H'} \Pr[h(i) \in S_i] \leq \ell^{-0.1k} \leq \delta .
\end{align*}
Second, Equation~\eqref{eq:head} implies
\begin{align*}
\Pr_{g}[\intr{i \in H'}g(i) \in S_i]  \leq \ell^{-0.1k}  + \eps \leq \delta .
\end{align*}

\paragraph{Fooling the Tail:}
We may assume that $q_i  \geq 1/m$ and $p_i > 0$ for all $i \in T$, 
since otherwise $S_i$ is trivial and we can drop such an index.
As in the proof of Lemma~\ref{lem:base}, by restricting to
a subset if necessary, we can also assume that 
\begin{align}
\label{eq:q}
\sum_{i \in T}q_i \leq C \log(1/\delta) .
\end{align}
For simplicity of notation, we denote $|T|$ by $n$. 
Therefore, $n \leq C m\log(1/\delta)$.

Let
$$Y(a,i) = \mathbbm{1}(g'(a,i)) - p_i .$$ 
Since $g'(a,i)$ is uniform over $S_i$, 
$$\E[Y(a,i)] = 0, \ \Var[Y(a,i)] = q_ip_i.$$
Write
\begin{align*}
\Pr_x[\intr{i \in T}g(i) \in S_i] = \prod_{i \in T} \left(p_i +
\frac{1}{\ell}\sum_{a=1}^\ell Y(a,i)\right)  .
\end{align*}
Define new random variables
\[ A_i = \frac{1}{\ell p_i}\sum_{a=1}^\ell Y(a,i).\]
so that
\begin{align}
\label{eq:tail2}
\Pr_x[\intr{i \in T} g(i) \in S_i] = \prod_{i=1}^n p_i (1 +
A_i) = \prod_{i=1}^n p_i \cdot \left(\sum_{i=0}^{n}S_i(A_1,\ldots,A_n)\right) .
\end{align}
For $k \leq n$, define 
\begin{align*} 
P_k(A)  = \prod_{i = 1}^n p_i \cdot \left(\sum_{i=0}^k S_i(A_1,\ldots,A_n)\right) . 
\end{align*}
We will show that $P_k(A)$ is a good approximation to $P_n(A)$ under 
$(O(k),\eps^{O(1)})$-wise independence, hence under both $\D$ and $\U$.

\begin{claim}
\label{clm:PnWUandD}
Both $|\E_{\D}[P_n(A) - P_k(A)]|$ and $|\E_{\U}[P_n(A) - P_k(A)]|$
are at most $O(\ell^{-0.2k})$.
\end{claim}

The claim completes the proof:
\begin{align*}
|\E_\D[P_n(A)] - \E_\U[P_n(A)]| & \leq |\E_\D[P_n(A)] - \E_\D[P_k(A)]|  +
|\E_\D[P_k(A)] - \E_\U[P_k(A)]| \\
& \phantom{=}  + |\E_\U[P_n(A)] - \E_\U[P_k(A)]| .
\end{align*} 
Bound the first and third terms by $O(\ell^{-0.2k})$ using the claim. 
Bound the second term as follows.
Since $k > 1$, for all $i$,
\begin{align*} 
\Var[A_i] & = \frac{1}{\ell^2p^2_i}\sum_{a=1}^\ell\Var[Y(a,i)] = 
\frac{q_i}{\ell p_i} \leq  \frac{q_i}{\ell^{0.9}} , \\
\nm(A_i) & \leq \frac{1}{\ell p_i}\sum_{a=1}^{\ell}\nm(Y(a,i)) \leq
\frac{2}{p_i} \leq \ell.
\end{align*}
Plugging in the bounds from Equations \eqref{eq:q}:
\begin{align*} 
\sum_{i=1}^n \Var[A_i] & \leq  
\frac{C \log(1/\delta)}{\ell^{0.9}} \leq \frac{1}{\ell^{0.6}}, \\
\sum_{i=1}^n \nm(A_i) & \leq C m \log(1/\delta) \ell \leq m \ell^{O(1)} , \\
\nm(S_k(A_1,\ldots,A_n)) & \leq \left(\sum_{i=1}^n \nm(A_i)\right)^k \leq m^k \ell^{O(k)} . 
\end{align*}
Thus,
\begin{align*}
|\E_\D[P_k(A)] - \E_\U[P_k(A)]| \leq \eps\nm(P_k) \leq \eps\ell^{O(k)}
= O(\ell^{-k}) .
\end{align*}
Overall,
\begin{align*}
  \Big|\Pr_g[\intr{i \in T}g(i) \in S_i] & - \Pr_h[\intr{i \in T}h(i)
  \in S_i] \Big| \\
  & = | \E_\D[P_n(A)] - \E_\U[P_n(A)]| \leq O(\ell^{-0.2k}) 
\leq \delta .
\end{align*}

\begin{proof}[Proof of Claim~\ref{clm:PnWUandD}]
We argue for $\D$, the same argument holds for $\U$.
Write
\begin{align*} 
|P_n(A) - P_k(A)| \leq \prod_{i=1}^n p_i
\cdot  \left|\sum_{i=k+1}^nS_i(A_1,\ldots,A_n)\right|.
\end{align*}
If  $A_1,\ldots,A_n$ are $(O(k),0)$-wise independent, then,
by Lemma \ref{lem:exp-kwise},
\[ \E[S_k(A_1,\ldots,A_n)^2] = \frac{(\sum_i\Var[A_i])^{k}}{k!} \leq \frac{\ell^{-0.6 k}}{k!} . \]
Hence, under $(O(k),\eps)$-wise independence,
\begin{align}
\label{eq:exp-bound}
\E[S_k(A_1,\ldots,A_n)^2] \leq \frac{\ell^{-0.6 k}}{k!} +
\eps\nm(S_k) \leq \frac{\ell^{-0.5 k}}{k!} .
\end{align}
We now repeat the proof of Lemma \ref{lem:main} with $\sigma^2 =
\ell^{-0.5}$ and $t = \ell^{0.2}$.
The event $G$ defined as
\begin{align*}
G = \left\{ |S_{k}(A)| \leq \frac{\ell^{-0.05 k}}{\sqrt{k!}} \ \text{and}
\ |S_{k+1}(A)| \leq \frac{\ell^{-0.05(k+1)}}{\sqrt{(k+1)!}}\right\} 
\end{align*}
occurs with probability at least $1-2\ell^{-0.4k}$. 
As in the proof of Theorem~\ref{thm:main}, conditioned on $G$,
\begin{align}
|P_n(A) - P_k(A)| \leq \sum_{i=k}^n |S_i(A)| & \leq 2 (6e\ell^{-0.05})^k. 
\label{eq:tail-sum}
\end{align}
Since $\E[A_i] =0$ for all $i$
and $\nm(S_k) \leq \ell^{O(k)}$, 
by Equation~\eqref{eq:exp-bound},
it follows that under $(O(k),\eps)$-wise independence,
\begin{align}
\label{eq:pk2}
\E[P_k(A_1,\ldots,A_n)^2 ] & \leq \sum_{i=0}^k\E[S_i(A_1,\ldots,A_n)^2] + \eps \ell^{O(k)} = O(1).
\end{align}

Denote by $\neg G$ the complement of $G$.
Write
\begin{align*} 
\E[P_n(A) - P_k(A)] = \E[(P_n(A) - P_k(A))\mathbbm{1}(G)] +
\E[(P_n(A) - P_k(A))\mathbbm{1}(\neg G)] .
\end{align*}
Equation \eqref{eq:tail-sum} implies
\begin{align*}
|\E[(P_n(A) - P_k(A))\mathbbm{1}(G)]| \leq 2(20\ell^{-0.25})^k .
\end{align*}

It remains to bound the second term.
Bound
\begin{align*}
|\E[(P_n(A) - P_k(A))\mathbbm{1}(\neg G)]| 
\leq |\E[P_n(A)\mathbbm{1}(\neg G)]| + |\E[P_k(A)\mathbbm{1}(\neg G)]|
\end{align*}
Note that $0 \leq P_n(A) \leq 1$ since it is the probability of an event.
Also note that $\E[P_k(A)^2] = O(1)$ by Equation \eqref{eq:pk2}. 
So we can bound the RHS using Holder's inequalities by
\begin{align*}
|\E[\mathbbm{1}(\neg G)]| + |\E[P_k(A)^2]^{1/2} \cdot \E[\mathbbm{1}(\neg
  G)]^{1/2}| \leq O((\E[\mathbbm{1}(\neg G)])^{1/2}) = O(\ell^{-0.2k}).
\end{align*}
\end{proof}
\end{proof}

We are ready to prove the main theorem of this section.

\begin{proof}
The proof uses an hybrid argument.
The $\gmr$ generator chooses $g_0:[n] \rgta [m_0]$, and then
$g'_1,\ldots, g'_T$ where $g'_t = [m_{t-1}] \times [n] \rgta [m_t]$
has error $\delta' = \delta/T$ and defines 
$$g_t(i) = g'_t(g_{t-1}(i),i).$$ 
Let $h_0, h'_1,\ldots, h'_t$ be truly random hash 
functions with similar domains and ranges.
For $0 \leq t,l \leq T$, define the hybrid family $\G^l_t =\{f^l_t:[m] \rgta [n]\}$
as follows: for $t=0$ and every $l$,
$$f^l_0 = h_0,$$ 
and for $t > 0$ and every $l$,
\[ f^l_t(i) = 
\begin{cases} 
g'_t(f^l_{t-1}(i),i) & \text{for} \  l < t  , \\
h'_t(f^l_{t-1}(i),i) & \text{for} \  t \leq l .
\end{cases}
\]
For every $l$, let $\G^l = \G^l_T$.
Thus, $\G^0 = \gmr$ and $\G^T = \uh$.
We will show by induction on $l \geq 1$ that
\[ \left|\Pr_{f^l \in \G^l}[\intr{i \in [n]} f^l(i) \in S_i] - 
\Pr_{f^{l-1} \in \G^{l-1}}[\intr{i \in [n]} f^{l-1}(i) \in S_i]\right| \leq \delta' . \]
The desired bound then follows by the triangle inequality. 

In the base case when $l=1$, couple $\G^0$ and $\G^1$ by
picking the same $g'_1,\ldots,g'_T$, and use them to define the
function $f':[m_1] \times [n] \rgta [m]$ so that
\[ f^0(i) = f'(g_0(i),i), \  f^1(i) = f'(h_0(i),i).\]
For $i \in [n]$, define 
$$S'_i = \{a\in [m_1]: f'(a,i) \in S_i \}.$$ 
Thus,
\begin{align*} 
\Big|\Pr_{f^1 \in \G^1}[\intr{i \in [n]}f^1(i) \in S_i] & - 
\Pr_{f^0 \in \G^0}[\intr{i \in [n]}f^0 (i) \in S_i]\Big| \\
&  = 
\left|\Pr_{h_0}[\intr{i \in [n]} h_0(i) \in S'_i] - 
\Pr_{g_0}[\intr{i \in [n]} g_0(i) \in S'_i]\right|  \leq \delta' , 
\end{align*}
by applying Lemma~\ref{lem:base} with 
$k =O(\log(1/\delta'))$ and $\eps = \delta' \cdot m_0^{-O(k)}$. 

For the inductive case $l > 1$, couple $\G^l$ and $\G^{l-1}$ by
picking the same $g'_{l+1},\ldots,g'_T$, 
and pick $x \in [m_{l-1}]^n$ uniformly at random.
There is a function $f':[m_l] \times [n] \rgta [m]$ so that
\[ f^l(i) = f'(h'_l(x_i,i),i), \  f^{l-1}(i) = f'(g'_l(x_i,i),i).\]
As before, define 
$$S_i = \{a \in [m_l]: f'(a,i) \in S_i\}.$$ 
Hence,
\begin{align*} 
\Big|\Pr_{f^l \in \G^l}[\intr{i \in [n]} f^l(i) \in S_i] & - 
\Pr_{f^{l-1} \in \G^{l-1}}[\intr{i \in [n]} f^{l-1}(i) \in S_i] \Big| \\
&  =  \Big|\Pr_{h_l,x}[\intr{i \in [n]} h'_l(x_i,i) \in S'_i] - 
\Pr_{g_l,x}[\intr{i \in [n]} g'_l(x_i,i) \in S'_i]\Big| \leq \delta ,
\end{align*}
by Lemma~\ref{lem:induct} with 
\begin{align*}  
k_{l-1} > 1, \
m_{l-1} \geq \log(1/\delta')^{C}, \ m_{l-1} \geq k_{l-1} , \ 
m_{l-1}^{-k} \leq {\delta'}^{C} , \ 
\eps_{l-1}  \leq (m_l m_{l-1})^{-C k} . 
\end{align*}
\end{proof}

\section*{Acknowledgements}

We thank Nati Linial, Raghu Meka, Yuval Peres, Dan Spielman, Avi Wigderson and David Zuckerman for
helpful discussions. We thank an anonymous referee for pointing out an error in 
the statement of Theorem \ref{thm:main} in a previous version of the paper.

\bibliographystyle{alpha}
\bibliography{references}

\newcommand{\etalchar}[1]{$^{#1}$}
\begin{thebibliography}{GMR{\etalchar{+}}12}

\bibitem[ABI86]{ABI}
Noga Alon, Laszlo Babai, and Alon Itai.
\newblock A fast and simple randomized parallel algorithm for the maximal
  independent set problem.
\newblock {\em J. Algorithms 7(4)}, 1986.

\bibitem[ASWZ96]{ArmoniSWZ96}
Roy Armoni, Michael~E. Saks, Avi Wigderson, and Shiyu Zhou.
\newblock Discrepancy sets and pseudorandom generators for combinatorial
  rectangles.
\newblock In {\em 37th Annual Symposium on Foundations of Computer Science,
  {FOCS} '96}, pages 412--421, 1996.

\bibitem[BCFM00]{BCFM}
Andrei~Z. Broder, Moses Charikar, Alan~M. Frieze, and Michael Mitzenmacher.
\newblock Min-wise independent permutations.
\newblock {\em J. Comput. Syst. Sci.}, 60(3):630--659, 2000.

\bibitem[BCM98]{BCM}
Andrei~Z. Broder, Moses Charikar, and Michael Mitzenmacher.
\newblock A derandomization using min-wise independent permutations.
\newblock In {\em Randomization and Approximation Techniques in Computer
  Science, Second International Workshop, RANDOM'98}, pages 15--24, 1998.

\bibitem[CG89]{CG}
Benny Chor and Oded Goldreich.
\newblock On the power of two-point based sampling.
\newblock {\em J. Complexity 5(1)}, 1989.

\bibitem[EGL{\etalchar{+}}98]{EGLNV}
Guy Even, Oded Goldreich, Michael Luby, Noam Nisan, and Boban Velickovic.
\newblock Efficient approximation of product distributions.
\newblock {\em Random Struct. Algorithms}, 13(1):1--16, 1998.

\bibitem[GKM15]{GKM15}
Parikshit Gopalan, Daniel Kane, and Raghu Meka.
\newblock Pseudorandomness via the discrete {F}ourier transform.
\newblock In {\em Accepted to IEEE FOCS 2015}, 2015.

\bibitem[GMR{\etalchar{+}}12]{GMRTV}
Parikshit Gopalan, Raghu Meka, Omer Reingold, Luca Trevisan, and Salil~P.
  Vadhan.
\newblock Better pseudorandom generators from milder pseudorandom restrictions.
\newblock In {\em 53rd Annual {IEEE} Symposium on Foundations of Computer
  Science, {FOCS'2012}}, pages 120--129, 2012.

\bibitem[Ind99]{Indyk}
Piotr Indyk.
\newblock A small approximately min-wise independent family of hash functions.
\newblock In {\em Proceedings of the Tenth Annual {ACM-SIAM} Symposium on
  Discrete Algorithms}, pages 454--456, 1999.

\bibitem[LLSZ97]{LinialLSZ97}
Nathan Linial, Michael Luby, Michael~E. Saks, and David Zuckerman.
\newblock Efficient construction of a small hitting set for combinatorial
  rectangles in high dimension.
\newblock {\em Combinatorica}, 17(2):215--234, 1997.

\bibitem[Lu02]{Lu02}
Chi{-}Jen Lu.
\newblock Improved pseudorandom generators for combinatorial rectangles.
\newblock {\em Combinatorica}, 22(3):417--434, 2002.

\bibitem[Mul96]{Mulmuley}
Ketan Mulmuley.
\newblock Randomized geometric algorithms and pseudorandom generators.
\newblock {\em Algorithmica}, 16(4/5):450--463, 1996.

\bibitem[NN93]{NaorN93}
Joseph Naor and Moni Naor.
\newblock Small-bias probability spaces: Efficient constructions and
  applications.
\newblock {\em SIAM J. on Comput.}, 22(4):838--856, 1993.

\bibitem[PS76]{PS}
G~Polya and G.~Szego.
\newblock {\em Problems and Theorems in Anaysis II}.
\newblock Springer Classics in Mathematics, 1976.

\bibitem[SSZZ99]{SaksSZZ99}
Michael~E. Saks, Aravind Srinivasan, Shiyu Zhou, and David Zuckerman.
\newblock Low discrepancy sets yield approximate min-wise independent
  permutation families.
\newblock In {\em Third International Workshop on Randomization and
  Approximation Techniques in Computer Science RANDOM'99}, pages 11--15, 1999.

\bibitem[Ste04]{CS}
J.~Michael Steele.
\newblock {\em The Cauchy-Schwarz Master Class}.
\newblock Cambridge University Press, 2004.

\end{thebibliography}

\appendix
\section{Missing Proofs}
\label{app:proof}

We give the proof of Fact \ref{fact:2} which states that over the reals, if $S_k(a) = S_{k+1}(a) = 0$ for $k > 0$
then $S_\ell(a) =0$ for all $\ell \geq k$.

For a univariate polynomial $p(\xi)$ and a root $y \in \R$ of $p$, 
denote by $\mult(p,y)$ the multiplicity of the root $y$ in $p$. 
We use the following property of polynomials $p(\xi)$
with real roots \cite{PS}, which can be proved using the interlacing of
the zeroes of $p(\xi)$ and $p'(\xi)$: If $\mult(p',y) \geq2$ 
then $\mult(p,y) \geq \mult(p',y) + 1$.

\begin{proof}[Proof of Fact \ref{fact:2}]
Let 
\[p(\xi) = \prod_{i \in [n]} (\xi + b_i) = \sum_{k=0}^n \xi^k
S_{n-k}(b).\]
Consider $p^{(n-k -1)}(\xi)$ which is the $(n - k-1)^{th}$ derivative of
$p(\xi)$. Since $S_k(b) = S_{k+1}(b) = 0$ for $k > 0$, 
it follows that $\xi^2$ divides $p^{(n-k-1)}(\xi)$ and hence 
$\mult(p^{(n-k-1)},0) \geq 2$. 
Applying the above fact $n - k-1$ times, we get $\mult(p,0) \geq n -k +1$
so $S_n(b) = \ldots = S_{k}(b) = 0$.
%
\end{proof}

The next Theorem is a routine extension of the result of \cite{SaksSZZ99} to large $\ell$.

\begin{theorem}\cite{SaksSZZ99}
\label{thm:prg-to-hash}
Let $\gcr:\zo^r \rgta [m]^n$ be a $\prg$ for combinatorial
rectangles with error $\eps$. The resulting family
$\{g_y : y \in \zo^r \}$ of hash functions is approximately $\ell$-minima-wise
independent with error at most $\eps {m \choose \ell}$.
\end{theorem}
\begin{proof}[Proof of Theorem \ref{thm:prg-to-hash}]
Fix $S\subseteq [n]$ and a sequence $T= (t_1,\ldots,t_\ell)$ of $\ell$
distinct elements from $S$. The event 
\[ g(t_1) < \cdots < g(t_\ell) < \min g(S\setminus T) \]
can be viewed as the disjoint union of ${m \choose \ell}$ events by
fixing the set $A = \{a_1 < \ldots < a_\ell\}$ that $T$ maps to. 
The indicator $\mathbbm{1}_A$ of the event
\[ g(t_1) = a_1, \ldots, \ g(t_\ell) = a_\ell, \ g(S\setminus T) > a_\ell \]
is a combinatorial rectangle: 
Define
\begin{align*} 
f_i(x_i) & = 1 \ \text{for} \ i \not\in S\\
f_i(x_i) & = \mathbbm{1}(x_i = a_j) \ \text{for} \ i = t_j \in T \\
f_i(x_i) & = \mathbbm{1}(x_i > a_\ell) \ \text{for} \ i \in S\setminus T
\end{align*}
and
\begin{align*}
f_A(x_1,\ldots,x_n) & = \andop{i \in [n]}f_i(x_i) .
\end{align*}
Since $g(i) = x_i$, it follows that $\mathbbm{1}_A(g) =
f_A(x)$. Further, choosing $h \in \uh$ is equivalent to choosing $x \in
[m]^n$ uniformly at random.  Hence, 
\begin{align*}
\Pr_{g \in \gcr} [g(t_1) < & \cdots < g(t_\ell) < \min g(S\setminus T)] \\
& = \sum_{A}\E_{y \in \zo^r}[f_A(\gcr(y))]\\
& = \sum_{A}(\E_{h \in \uh}[\mathbbm{1}_A(h)] \pm \eps)\\
& = \Pr_{h \in \uh}[h(t_1) < \cdots < h(t_\ell) < \min h(S\setminus
  T)] \pm {m \choose \ell}\eps .
\end{align*}
\end{proof}

Finally we discuss how to generate the $(k,\eps)$-wise independent distributions on $[m]^n$
with seed length $O(\log\log(n) + k\log(m) + \log(1/\eps))$. We claim
that it suffices to take a $k' = k\log(m)$-wise $\eps$-independent string
of length $n' = n\log(m)$. Naor and Naor \cite{NaorN93}
showed that such distributions can be generated using seed-length
$O(\log\log(n) + k\log(m) + \log(1/\eps))$. 
We can also assume that every co-ordinate is uniformly random in
$[m]$ by adding the string $(a,\ldots,a)$ where $a \in [m]$ is chosen randomly.

%
%
%
%


\end{document}